%% file: main.tex
\documentclass[letterpaper,10pt,conference]{ieeeconf}

\IEEEoverridecommandlockouts 
\include{header_ieeeconf}

\newcommand{\con}{\ensuremath{\Ccal}}
\newcommand{\ass}{\ensuremath{\Acal}}
\newcommand{\gar}{\ensuremath{\Gcal}}
\newcommand{\env}{\ensuremath{\Ecal}}

\DeclareMathOperator{\comp}{\otimes}

\DeclareMathOperator{\image}{im}
\DeclareMathOperator{\kernel}{ker}
\renewcommand{\star}{*}

\begin{document}

\title{Contracts as specifications for dynamical systems in driving variable form}

\author{Bart Besselink,
        Karl H.\ Johansson,
        Arjan van der Schaft%
\thanks{Bart Besselink and Arjan van der Schaft are with the Jan C.\ Willems Centre for Systems and Control and the Bernoulli Institute for Mathematics, Computer Science, and Artificial Intelligence, University of Groningen, The Netherlands; Email: b.besselink@rug.nl; a.j.van.der.schaft@rug.nl}%
\thanks{Karl H.\ Johansson is with the School of Electrical Engineering and Computer Science, KTH Royal Institute of Technology, Sweden; Email: kallej@kth.se}}

\maketitle
\thispagestyle{empty}
\pagestyle{empty}

\begin{abstract}
This paper introduces assume/guarantee contracts on continuous-time control systems, hereby extending contract theories for discrete systems to certain new model classes and specifications. Contracts are regarded as formal characterizations of control specifications, providing an alternative to specifications in terms of dissipativity properties or set-invariance. The framework has the potential to capture a richer class of specifications more suitable for complex engineering systems. The proposed contracts are supported by results that enable the verification of contract implementation and the comparison of contracts. These results are illustrated by an example of a vehicle following system.
\end{abstract}

\section{Introduction}\label{sec_introduction}
Specifications on dynamical (control) systems typically come in the form of requirements on stability, performance (generally expressed as a bounded gain for a suitably chosen input-output pair), or passivity. Such specifications have in common that they can be captured in the elegant framework of dissipativity as introduced in \cite{willems_1972}, see also \cite{book_vanderschaft_2017,book_arcak_2016} and~\cite{book_sepulchre_1997} for related control approaches.
An alternative class of specifications can be characterized through set-invariance techniques, capturing properties such as safety, e.g., \cite{blanchini_1999}.

However, stricter performance requirements and increasing complexity of modern engineering systems such as intelligent transportation systems or smart manufacturing systems require the expression of control specifications that go beyond dissipativity or invariance. This observation motivates the work on formal methods in control, e.g., \cite{book_tabuada_2009,book_belta_2017}, which generally requires the abstraction of continuous dynamical systems to discrete transition systems because of the need to express logic specifications such as LTL.

A different approach is taken in this paper. Namely, we present an approach for expressing rich specifications on dynamical systems directly in the continuous domain by introducing \emph{contracts} for linear time-invariant dynamical systems, leading to the following contributions.

First, inspired by contracts for discrete systems in computer science developed in~\cite{benveniste_2008} and~\cite{benveniste_2018}, we define \emph{assume/guarantee contracts} for a class of continuous-time linear dynamical systems. A contract is a pair of dynamical systems known as the assumptions and guarantees and can be regarded as a specification on the dynamical system. Namely, a system implements a contract (i.e., satisfies the specification) when it satisfies the guarantees whenever it is interconnected with an environment that satisfies the assumptions; a definition that will be made precise by using the notion of \emph{simulation} (see, e.g., \cite{pappas_2003,vanderschaft_2004}) as a means for comparing system behavior.

Second, we present a result that allows for efficiently verifying whether a dynamical system implements a given contract. In particular, geometric conditions are given for contract implementation enabling the use of tools from geometric control theory, e.g., \cite{book_wonham_1979,book_basile_1992,book_trentelman_2001}.

Third, a notion of contract \emph{refinement} is developed as a means for comparing contracts, i.e., providing a way to formalize whether a given contract provides tighter or relaxed requirements with respect to a second contract. The definition is again inspired by results in computer science, see \cite{bauer_2012,benveniste_2018}.

Fourth, the contract approach is illustrated by application to a vehicle following system as used for, e.g., vehicle platooning \cite{alam_2015b}. Here, the objective is to \emph{guarantee} a desired inter-vehicle distance under the \emph{assumption} that the lead vehicle satisfies the kinematic relation. Crucially, knowledge of the exact dynamics of the lead vehicle is not assumed.

Through the above contributions, it is argued that assume/guarantee contracts provide various distinguishing and useful features with respect to specifications expressed using dissipativity theory or set invariance. Namely, the explicit characterization (through the assumptions) of the set of environments in which the dynamical system is expected to operate potentially allows for relaxing requirements on the system (as the guarantees might be partially ensured by the assumptions). The characterization of these environments is particularly relevant in the analysis of interconnected systems; an important topic that will be explored in future work.

Moreover, by using simulation relations as a basis for comparing system behavior (recall that the assumptions and guarantees are themselves dynamical systems), a rich class of system behaviors can be characterized including dynamic behavior. We note that the static nature of the supply rates limits the ability of traditional dissipativity theory to capture dynamic behavior. Dynamic supply rates \cite{book_arcak_2016} and integral quadratic constraints \cite{megretski_1997} have been introduced to address this limitation, but these approaches do again not characterize the environment in which a system operates.

Related work on contracts for dynamical systems is presented in \cite{saoud_2018}, where contracts are used to capture set-invariance properties in input and output spaces. As such, the contracts in \cite{saoud_2018} do not allow for including dynamics in the specification. Richer contracts, called parametric assume/guarantee contracts, are defined in~\cite{kim_2017}, allowing for expressing input-output gain properties. However, only discrete-time systems are considered in this work.

The remainder of this paper is organized as follows. Section~\ref{sec_dvsystems} introduces the class of systems considered in this work and develops a notion of simulation for such systems. Corresponding compositional properties are given in Section~\ref{sec_composition} before developing contracts as system specifications in Section~\ref{sec_contracts}. Section~\ref{sec_example} presents an illustrative example and Section~\ref{sec_conclusions} concludes the paper.

\textit{Notation.} For a linear map $A:\Xcal\rightarrow\Ycal$ with $\Xcal$ and $\Ycal$ finite-dimensional vector spaces, $\image A$ and $\kernel A$ denote the image and kernel of $A$, respectively. Given a linear subspace $\Vcal\subset\Xcal\times\Ycal$, let $\pi_{\Xcal}(\Vcal) = \{ x \:|\: \exists y \text{ s.t.\ } (x,y)\in\Vcal \}$ be the projection of $\Vcal$ on $\Xcal$; $\pi_{\Ycal}(\Vcal)$ is defined similarly.

\section{Systems in driving variable form}\label{sec_dvsystems}
Consider the linear dynamical system
\begin{align}
\sys_i:\left\{\begin{array}{rl}
\dot{x}_i &= A_ix_i + G_id_i, \\ w_i &= C_ix_i, \\ 0 &= H_ix_i,
\end{array}\right.\label{eqn_sysi}
\end{align}
with state $x_i\in\Xcal_i$, external variable $w_i\in\Wcal$, and driving variable $d_i\in\Dcal_i$. Here, $\Xcal_i$, $\Wcal$, and $\Dcal_i$ are finite-dimensional vector spaces. The system (\ref{eqn_sysi}) is regarded as an \emph{open} system in which the external variable $w_i$ interacts with the environment, whereas the state $x_i$ is the internal variable. The driving variable $d_i$ plays the role of generator of trajectories.
\begin{remark}
The interpretation of the system (\ref{eqn_sysi}) in terms of external variables and (internal) state variables is similar to the perspective taken in the behavioral approach to system theory, e.g., \cite{willems_2007b}. In fact, the form (\ref{eqn_sysi}) (but without constraints $H_ix_i=0$) is given in \cite{willems_1983} as one representation of this perspective. We stress that no explicit distinction is made between inputs and outputs in (\ref{eqn_sysi}), even though the external variable $w_i$ could be partitioned as such (see again~\cite{willems_2007b}).
\end{remark}
\begin{remark}
The algebraic constraints in (\ref{eqn_sysi}) provide a flexible system description that will turn out to be useful in defining system composition (in Section~\ref{sec_composition}) as well as in formalizing complex specification, see the example in Section~\ref{sec_example}.
\end{remark}

Due to the algebraic constraints in (\ref{eqn_sysi}), not all initial conditions lead to feasible trajectories. This motivates the introduction of the \emph{consistent subspace} $\Vcal_i^{\star}$ as the set of initial conditions $x_i(0)$ for which there exists (for some $d_i(\cdot)$) a trajectory $x_i(\cdot)$ that satisfies the constraints, i.e., $H_ix_i(t) = 0$ for all $t\in\R_+$. The consistent subspace can be characterized as the largest (with respect to subspace inclusion) subspace $\Vcal_i\subset\Xcal_i$ such that
\begin{align}
A_i\Vcal_i \subset \Vcal_i + \image G_i, \qquad \Vcal_i \subset \kernel H_i,
\label{eqn_consistentsubspace}
\end{align}
see, e.g., \cite{vanderschaft_2004b,megawati_2016}.

Following results on unconstrained systems in \cite{pappas_2003,vanderschaft_2004} and constrained systems in \cite{vanderschaft_2004b}, the notion of \emph{simulation relation} is introduced as a means for comparing the behavior of systems $\sys_1$ and $\sys_2$. Such system comparison will turn out to be crucial for expressing rich system specifications.
\begin{definition}\label{def_simulationrel}
A linear subspace $\Scal\subset\Xcal_1\times\Xcal_2$ satisfying $\pi_{\Xcal_i}(\Scal) \subset \Vcal_i^{\star}$, $i\in\{1,2\}$, is a simulation relation of $\sys_1$ by $\sys_2$ if, for all $(x_1(0),x_2(0))\in\Scal$, the following hold:
\begin{enumerate}
	\item\label{def_simulationrel_item1} for each driving function $d_1(\cdot)$ such that the corresponding state trajectory $x_1(\cdot)$ with initial condition $x_1(0)$ satisfies $x_1(t)\in\Vcal_1^{\star}$ for all $t\in\R_+$, there exists a driving function $d_2(\cdot)$ such that the corresponding state trajectory $x_2(\cdot)$ with initial condition $x_2(0)$ satisfies
	\begin{align}
	\big(x_1(t),x_2(t)\big) \in \Scal \label{eqn_def_simulationrel_state}
	\end{align}
	for all $t\in\R_+$;
	\item\label{def_simulationrel_item2} the external variables are equal, i.e.,
	\begin{align}
	C_1x_1(0) = C_2x_2(0). \label{eqn_def_simulationrel_output}
	\end{align}
	\vskip1mm%
\end{enumerate}
\end{definition}

Whereas simulation relations are defined in terms of system trajectories, equivalent algebraic conditions can be obtained using standard arguments from geometric control theory, e.g., \cite{book_wonham_1979,book_basile_1992,book_trentelman_2001}. This is formalized next.
\begin{lemma}\label{lem_simulationrel}
A linear subspace $\Scal\subset\Xcal_1\times\Xcal_2$ satisfying $\pi_{\Xcal_i}(\Scal)\subset\Vcal_i^{\star}$, $i\in\{1,2\}$, is a simulation relation of $\sys_1$ by $\sys_2$ if and only if the following conditions hold for all $(x_1,x_2)\in\Scal$:
\begin{enumerate}
	\item\label{lem_simulationrel_item1} for all $d_1\in\Dcal_1$ such that $A_1x_1 + G_1d_1\in\Vcal_1^{\star}$, there exists $d_2\in\Dcal_2$ such that $A_2x_2 + G_2d_2\in\Vcal_2^{\star}$ and
	\begin{align}
	\big(A_1x_1 + G_1d_1,A_2x_2 + G_2d_2\big)\in\Scal;
	\label{eqn_lem_simulationrel_state}
	\end{align}
	\item\label{lem_simulationrel_item2} the external variables are equal, i.e.,
	\begin{align}
	C_1x_1 = C_2x_2. \label{eqn_lem_simulationrel_output}
	\end{align}
	\vskip1mm%
\end{enumerate}
\end{lemma}
\begin{proof}
The proof follows that of \cite[Proposition~3.1]{megawati_2016}, see also \cite[Proposition~2.9]{vanderschaft_2004} for unconstrained systems.
\end{proof}

Now, \emph{simulation} can be defined directly using the notion of simulation relation in Definition~\ref{def_simulationrel}.
\begin{definition}\label{def_simulation}
A system $\sys_1$ is said to be simulated by $\sys_2$ (or, $\sys_2$ simulates $\sys_1$), denoted $\sys_1\preccurlyeq\sys_2$, if there exists a simulation relation $\Scal$ of $\sys_1$ by $\sys_2$ satisfying $\pi_{\Xcal_1}(\Scal) = \Vcal_1^{\star}$ and $\pi_{\Xcal_2}(\Scal)\subset\Vcal_2^{\star}$. A simulation relation with this property will be referred to as a full simulation relation.
\end{definition}

The following lemma gives an important property of the notion of simulation, which enables its use as a means for comparing the behavior of systems $\sys_i$. The result of this lemma will be exploited later.
\begin{lemma}\label{lem_simulationpreorder}
Simulation $\preccurlyeq$ is a preorder. Namely, it is
\begin{enumerate}
	\item\label{lem_simulationpreorder_transitivity} reflexive, i.e., $\sys_1\preccurlyeq\sys_1$ for any $\sys_1$;
	\item transitive, i.e., for any systems $\sys_i$, $i\in\{1,2,3\}$ satisfying $\sys_1\preccurlyeq\sys_2$ and $\sys_2\preccurlyeq\sys_3$, it holds that $\sys_1\preccurlyeq\sys_3$.
\end{enumerate}
\end{lemma}
\begin{proof}
Reflexivity of the simulation operation follows by choosing $\Scal = \{(x_1,x_1)\:|\: x_1\in\Vcal^{\star}_1\}$, which is easily verified to be a full simulation relation of $\sys_1$ by itself. To prove transitivity, let $\Scal_{12}$ and $\Scal_{23}$ be (full) simulation relations of $\sys_1$ by $\sys_2$ and $\sys_2$ by $\sys_3$, respectively. Following~\cite{kerber_2010}, define
\begin{align}
\Scal_{13} = \big\{ (x_1,x_3) \:\big|\: &\exists x_2\in\Xcal_2 \text{ such that } \nonumber\\
&(x_1,x_2)\in\Scal_{12}, (x_2,x_3)\in\Scal_{23} \big\}.
\end{align}
Then, it can be checked that $\Scal_{13}$ defines a full simulation relation of $\sys_1$ by $\sys_3$.
\end{proof}

\section{System composition}\label{sec_composition}
In this section, we consider the interconnection of systems through their external variables. Given two systems $\sys_i$, $i\in\{1,2\}$, of the form (\ref{eqn_sysi}), their composition $\sys_1\comp\sys_2$ is defined as the system resulting from setting
\begin{align}
w_1 = w_2.
\label{eqn_systemcomposition_variablesharing}
\end{align}
Thus, system composition is regarded as \emph{variable sharing}; a perspective that is advocated in, e.g., \cite{willems_2007b}.

Following (\ref{eqn_systemcomposition_variablesharing}), a realization of $\sys_1\comp\sys_2$ is given as
\begin{align}
\sys_1\comp\sys_2:\left\{\begin{array}{rl}
\dot{x}^{\comp} &= A^{\comp}x^{\comp} + G^{\comp}d^{\comp}, \\
w^{\comp} &= C^{\comp}x, \\
0 &= H^{\comp}x,
\end{array}\right.\label{eqn_comp_sys}
\end{align}
with state $x^{\comp} = (x^{\comp}_1,x^{\comp}_2)\in\Xcal_1\times\Xcal_2$, external variable $w^{\comp}\in\Wcal$, and driving variable $d^{\comp} = (d^{\comp}_1,d^{\comp}_2)\in\Dcal_1\times\Dcal_2$. The linear maps in (\ref{eqn_comp_sys}) are given by
\begin{alignat}{2}
A^{\comp} &= \left[\begin{array}{cc} A_1 & 0 \\ 0 & A_2 \end{array}\right], &
G^{\comp} &= \left[\begin{array}{cc} G_1 & 0 \\ 0 & G_2 \end{array}\right], \\
C^{\comp} &= \tfrac{1}{2}\!\left[\begin{array}{cc} C_1 & C_2 \end{array}\right], \quad &
H^{\comp} &= \left[\begin{array}{cc} H_1 & 0 \\ 0 & H_2 \\ C_1 & -C_2 \end{array}\right].
\label{eqn_comp_matrices}
\end{alignat}
Note that the final constraint imposed by $H^{\comp}$ restricts the external behavior of (\ref{eqn_comp_sys}) to the external behavior that is common to $\sys_1$ and $\sys_2$, which agrees with (\ref{eqn_systemcomposition_variablesharing}). Due to this additional constraint, the consistent subspace $\Vcal^{\comp,\star}$ of $\sys_1\comp\sys_2$ can not easily be expressed in terms of those of $\sys_1$ and $\sys_2$. Instead, recall that $\Vcal^{\comp,\star}$ is the largest subspace $\Vcal^{\comp}$ satisfying
\begin{align}
A^{\comp}\Vcal^{\comp} \subset \Vcal^{\comp} + \image G^{\comp}, \qquad \Vcal^{\comp}\subset \kernel H^{\comp}.
\label{eqn_comp_consistentsubspace}
\end{align}

The next result states that the composed system $\sys_1\comp\sys_2$ is simulated by both $\sys_1$ and $\sys_2$. In fact, it is the \emph{largest} (with respect to simulation) system with this property.
\begin{theorem}\label{thm_comp}
Consider systems $\sys_i$, $i\in\{1,2\}$, of the form~(\ref{eqn_sysi}) and let their composition $\sys_1\comp\sys_2$ be defined as in~(\ref{eqn_comp_sys})--(\ref{eqn_comp_matrices}). Then, the following two statements hold:
\begin{enumerate}
	\item\label{thm_comp_item1} For $i\in\{1,2\}$, $\sys_1\comp\sys_2$ is simulated by $\sys_i$, i.e.,
	\begin{align}
	\sys_1\comp\sys_2\preccurlyeq \sys_i.
	\label{eqn_thm_comp_upperbound}
	\end{align}
	\item\label{thm_comp_item2} Let $\sys$ be a system of the form (\ref{eqn_sysi}) that is simulated by both $\sys_1$ and $\sys_2$. Then, it is also simulated by $\sys_1\comp\sys_2$. Stated differently, the following holds:
	\begin{align}
	\sys \preccurlyeq \sys_i,\; i\in\{1,2\} \quad\implies\quad \sys \preccurlyeq \sys_1\comp\sys_2.
	\label{eqn_thm_comp_infimum}
	\end{align}
	\vskip1mm%
\end{enumerate}
\end{theorem}
\begin{proof}
The proof can be found in Appendix~A.
\end{proof}

Theorem~\ref{thm_comp} thus formalizes the intuition that the interconnection of systems through variable sharing in (\ref{eqn_systemcomposition_variablesharing}) can only restrict the behavior of systems (the first statement of the theorem). Another important consequence of Theorem~\ref{thm_comp} is that the property of simulation is preserved under system composition, as stated next.
\begin{theorem}\label{thm_compsim}
Let $\sys_i$ and $\sys_i'$, $i\in\{1,2\}$, be systems of the form (\ref{eqn_sysi}) such that $\sys_1\preccurlyeq\sys_1'$ and  $\sys_2\preccurlyeq\sys_2'$. Then,
\begin{align}
\sys_1\comp\sys_2 \preccurlyeq \sys_1'\comp\sys_2'.
\label{eqn_thm_compsim_simulationcomp}
\end{align}
\vskip1mm%
\end{theorem}
\begin{proof}
From the first statement in Theorem~\ref{thm_comp} we obtain
\begin{align}
\sys_1\comp\sys_2 \preccurlyeq \sys_i \preccurlyeq \sys_i',
\end{align}
for $i\in\{1,2\}$ and where the final simulation relation follows from the assumption $\sys_i\preccurlyeq\sys_i'$. Then, transitivity of the simulation relation (see Lemma~\ref{lem_simulationpreorder}) gives $\sys_1\comp\sys_2\preccurlyeq\sys_i'$, after which the application of the second statement of Theorem~\ref{thm_comp} leads to~(\ref{eqn_thm_compsim_simulationcomp}).
\end{proof}

\section{Contracts as specifications}\label{sec_contracts}
Consider the linear dynamical system
\begin{align}
\sys:\left\{\begin{array}{rl}
\dot{x} &= Ax + Gd, \\ w &= Cx, \\ 0 &= Hx,
\end{array}\right.
\label{eqn_sys}
\end{align}
as in (\ref{eqn_sysi}), but with indices removed for ease of presentation. The system (\ref{eqn_sys}) can interact with its environment through the external variable~$w$. To make this explicit, an environment $\env$ is defined to be a system of the same form, i.e.,
\begin{align}
\env:\left\{\begin{array}{rl}
\dot{x}^e &= A^ex^e + G^ed^e, \\ w^e &= C^ex^e, \\ 0 &= H^ex^e,
\end{array}\right.
\label{eqn_sysenv}
\end{align}
with $x^e\in\Xcal^e$ and driving variable $d^e\in\Dcal^e$. Finally, its external variables $w^e$ take values in the same space $\Wcal$ as the system (\ref{eqn_sys}), i.e., $w^e\in\Wcal$. Consequently, the interconnection of $\sys$ and $\env$ can be considered by setting $w = w^e$, leading to the system $\env\comp\sys$ with external variables $w\in\Wcal$.

We are interested in guaranteeing properties of $\sys$ when interconnected with relevant environments $\env$. To make this explicit, two systems will be introduced. First, define \emph{assumptions} $\ass$ as the system
\begin{align}
\ass:\left\{\begin{array}{rl}
\dot{x}^a &= A^ax^a + G^ad^a, \\ w^a &= C^ax^a, \\ 0 &= H^ax^a,
\end{array}\right.
\label{eqn_sysass}
\end{align}
with $x^a\in\Xcal^a$, $d^a\in\Dcal^a$, and external variables $w^a\in\Wcal$. Next, \emph{guarantees} $\gar$ are defined as
\begin{align}
\gar:\left\{\begin{array}{rl}
\dot{x}^g &= A^gx^g + G^gd^g, \\ w^g &= C^gx^g, \\ 0 &= H^gx^g,
\end{array}\right.
\label{eqn_sysgar}
\end{align}
where $x^g\in\Xcal^g$, $d^g\in\Dcal^g$, and $w^g\in\Wcal$. Note that both the assumptions~(\ref{eqn_sysass}) and guarantees~(\ref{eqn_sysgar}) are systems of the same form as $\sys$ in (\ref{eqn_sys}) and that they share the same space of external variables. Finally, $\Xcal^a$, $\Dcal^a$, $\Xcal^g$, and $\Dcal^g$ above are all finite-dimensional vector spaces.

The introduction of $\ass$ and $\gar$ allows for defining contracts.
\begin{definition}\label{def_contract}
A contract $\con$ is a pair of systems $(\ass,\gar)$ as in (\ref{eqn_sysass}) and (\ref{eqn_sysgar}).
\end{definition}

The relevance of contracts is given by their use as formal specifications for systems $\sys$ as in (\ref{eqn_sys}). This is made explicit using the following definition.
\begin{definition}\label{def_contractimplementation}
An environment $\env$ as in (\ref{eqn_sysenv}) is said to be compatible with the contract $\con = (\ass,\gar)$ if it is simulated by the assumptions $\ass$, i.e., $\env\preccurlyeq\ass$. A system $\sys$ as in (\ref{eqn_sys}) is said to be an implementation of the contract $\con = (\ass,\gar)$ if the composition of $\sys$ with any compatible environment is simulated by the guarantees $\gar$, i.e.,
\begin{align}
\env\comp\sys \preccurlyeq \gar
\end{align}
for all $\env\preccurlyeq\ass$.
\end{definition}

A contract $\con$ thus gives a formal specification for the external behavior of a system $\sys$ through two aspects. First, it specifies (by the assumptions $\ass$) the class of environments in which the system is supposed to operate. Second, it characterizes the required behavior of $\sys$ through the guarantees $\gar$, which the system needs to satisfy \emph{for any} compatible environment.
\begin{remark}
Whereas the concept of contracts was originally proposed in the scope of software engineering in \cite{meyer_1992}, Definition~\ref{def_contract} is inspired by assume/guarantee contracts in formal methods developed in \cite{benveniste_2008}, see also the recent book \cite{benveniste_2018} for a detailed discussion. We note that these works consider models of computation that are inherently discrete in nature, such that the theory developed in \cite{benveniste_2008,benveniste_2018} is not applicable to continuous dynamical systems as in (\ref{eqn_sys}).
\end{remark}

Whereas Definition~\ref{def_contractimplementation} defines contract implementation using a class of environments, the verification of contract implementation can be performed on the basis of the contract $\con = (\ass,\gar)$ directly, i.e., without explicitly constructing all compatible environments. This is stated next.
\begin{lemma}\label{lem_contractimplementation}
Consider a system $\sys$ as in (\ref{eqn_sys}) and a contract $\con = (\ass,\gar)$. Then, $\sys$ is an implementation of the contract if and only if
\begin{align}
\ass\comp\sys \preccurlyeq \gar.
\label{eqn_lem_contractimplementation}
\end{align}
\vskip1mm%
\end{lemma}
\begin{proof}
This is a direct result of Theorem~\ref{thm_comp} and the fact that $\env = \ass$ is a compatible environment.
\end{proof}

An important consequence of Lemma~\ref{lem_contractimplementation} is that it allows for efficiently verifying whether a system $\sys$ implements a given contract $\con$. To do so, the following theorem is instrumental.
\begin{theorem}\label{thm_contractimplementation}
Consider a system $\sys$ as in (\ref{eqn_sys}) and a contract $\con = (\ass,\gar)$. Let the consistent subspaces of $\ass\comp\sys$ and $\gar$ be denoted as $\Vcal^{\comp,\star}$ and $\Vcal^{g,\star}$, respectively. Then, a linear subspace $\Scal\subset\Xcal\times\Xcal^a\times\Xcal^g$ satisfies $\pi_{\Xcal^a\times\Xcal}(\Scal)\subset\Vcal^{\comp,\star}$ and $\pi_{\Xcal^g}(\Scal)\subset\Vcal^{g,\star}$ and is a simulation relation of $\ass\comp\sys$ by $\gar$ if and only if
\begin{align}
\left[\begin{array}{cc} A^{\comp} & 0 \\ 0 & A^g \end{array}\right]\Scal &\subset \Scal 
+ \image\left[\begin{array}{cc} G^{\comp} & 0 \\ 0 & G^g \end{array}\right], \label{eqn_thm_contractimplementation_state}\\
\left[\begin{array}{c} \image G^{\comp}\cap\Vcal^{\comp,\star} \\ 0 \end{array}\right] &\subset \Scal 
+ \left[\begin{array}{c} 0 \\ \image G^g\cap\Vcal^{g,\star} \end{array}\right], \label{eqn_thm_contractimplementation_input}\\
\Scal &\subset \kernel\left[\begin{array}{ccc} H^{\comp} & 0 \\ 0 & H^g \\ C^{\comp} & -C^g \end{array}\right],
\label{eqn_thm_contractimplementation_output}
\end{align}
where the linear maps $A^{\comp}$, $G^{\comp}$, $C^{\comp}$, $H^{\comp}$ form a realization of $\ass\comp\sys$. The system $\sys$ implements the contract $\con$ if and only if there exists a linear subspace $\Scal$ satisfying the above and, in addition, $\pi_{\Xcal^a\times\Xcal}(\Scal) = \Vcal^{\comp,\star}$.
\end{theorem}
\begin{proof}
The proof is given in Appendix~B.
\end{proof}
\begin{remark}\label{rem_geometriccontrol}
Theorem~\ref{thm_contractimplementation} enables the efficient algorithmic verification of contract implementation through the use of tools from geometric control theory, e.g., \cite{book_wonham_1979,book_trentelman_2001}. Namely, the so-called \emph{invariant subspace algorithm} can compute the largest (in the sense of subspace inclusion) subspace $\Scal$ that satisfies (\ref{eqn_thm_contractimplementation_state}) and (\ref{eqn_thm_contractimplementation_output}). For this subspace $\Scal^{\star}$, the condition (\ref{eqn_thm_contractimplementation_input}) as well as $\pi_{\Xcal\times\Xcal^a}(\Scal) = \Vcal^{\comp,\star}$ then need to be verified, which can again be done algorithmically. For an example of the use of the invariant subspace algorithm for system (bi)simulation (albeit for a different class of systems than the one studied in this paper), see \cite{vanderschaft_2004}.
\end{remark}
\begin{remark}\label{rem_contractequivalence}
From condition (\ref{eqn_lem_contractimplementation}) in Lemma~\ref{lem_contractimplementation}	and Theorem~\ref{thm_comp}, it can be concluded that if $\sys$ is an implementation of $\con = (\ass,\gar)$, it is also an implementation of $(\ass,\ass\comp\gar)$ (and vice versa). There is thus no restriction in replacing $\gar$ by $\gar' = \ass\comp\gar$.
\end{remark}

A distinguishing feature of using contracts as specifications is that contracts themselves can be compared through a notion of \emph{refinement} (see \cite{bauer_2012} for a similar definition in a more abstract setting).
\begin{definition}\label{def_contractrefinement}
A contract $\con' = (\ass',\gar')$ is said to refine a contract $\con = (\ass,\gar)$, denoted as $\con'\preccurlyeq\con$, if the following two conditions hold:
\begin{align}
\ass \preccurlyeq \ass', \qquad \ass\comp\gar' \preccurlyeq \gar.
\label{eqn_def_contractrefinement}
\end{align}
\vskip1mm%
\end{definition}

The above definition allows one to reason about tightening or relaxing specifications, where we note that this involves two aspects. Namely, a contract $\con'$ refines $\con$ if it simultaneously \emph{enlarges} the class of environments and asks for \emph{tighter} guarantees. This observation is made explicit as follows.
\begin{theorem}\label{thm_contractrefinement}
Let $\con' = (\ass',\gar')$ and $\con = (\ass,\gar)$ be contracts such that $\con'\preccurlyeq\con$. Then, the following holds:
\begin{enumerate}
	\item If $\env$ is a compatible environment for $\con$, then it is also a compatible environment for $\con'$.
	\item If $\sys$ is an implementation of $\con'$, then it is also an implementation of $\con$.
\end{enumerate}
\end{theorem}
\begin{proof}
Statement 1 is a direct result of the condition $\ass\preccurlyeq\ass'$ in (\ref{eqn_def_contractrefinement}), as transitivity of the simulation relation (see Theorem~\ref{lem_simulationpreorder}) gives that $\env\preccurlyeq\ass$ implies $\env\preccurlyeq\ass'$ (recall Definition~\ref{def_contractimplementation} on compatibility of an environment).

To prove statement 2, Let $\sys$ be an implementation of $\con'$. Now, consider
\begin{align}
\ass\comp\sys \preccurlyeq \ass'\comp\sys \preccurlyeq \gar'.
\end{align}
Here, the first simulation relation is the result of Theorem~\ref{thm_comp} and $\ass\preccurlyeq\ass'$, whereas the second follows from the assumption that $\sys$ implements $\con$, i.e., $\ass'\comp\sys\preccurlyeq\gar'$ by Lemma~\ref{lem_contractimplementation}. Using the definition of system composition and Theorem~\ref{thm_comp}, we also have $\ass\comp\sys\preccurlyeq\ass$, after which the same theorem gives $\ass\comp\sys \preccurlyeq \ass\comp\gar'$. The result then follows from (\ref{eqn_def_contractrefinement}) and transitivity of the simulation relation.
\end{proof}
\begin{remark}
When regarded as a means for characterizing control specifications, contracts in Definition~\ref{def_contract} provide an alternative to specifications expressed as dissipativity or set-invariance properties, see \cite{willems_1972,book_vanderschaft_2017,blanchini_1999}. Contracts have the unique feature that they explicitly characterize the set of environments in which a system $\sys$ is supposed to operate. In addition, the expression of contracts as (a pair of) dynamical systems and the use of the notion of simulation for comparing system behavior allows for defining specifications in which dynamic behavior can explicitly be taken into account, potentially allowing for expressing richer specifications than in dissipativity or set-invariance theory.
\end{remark}

\begin{figure}
\begin{center}
	\vskip1.5mm%
	\begin{tikzpicture}[
	vehicle/.style={rectangle, align=center, inner sep=0mm},
	signal/.style={color=black, line width=0.7pt, >=latex},
	gap/.style={semithick, >=stealth},
	sys/.style={rectangle, draw=black, line width=0.5pt, fill=kthblue!25,
		inner sep=0pt, minimum width=18mm, minimum height=0.9cm,
		text height=1.8ex,text depth=.25ex},			
	]
	\node (v1) at ( 0mm,0mm) [vehicle] {\includegraphics[width=25mm]{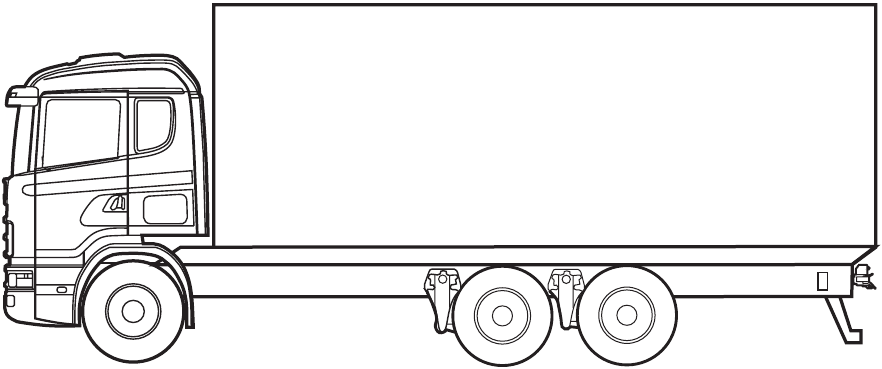}};	     
	\node (v2) at (45mm,0mm) [vehicle] {\includegraphics[width=25mm]{scania_small.pdf}};
	\draw[->,signal] ([yshift=1.5mm]v1.south west) -- node[pos=1,anchor=south,yshift=1mm]{$s_1$, $v_1$} +(-6mm,0mm);
	\draw[->,signal] ([yshift=1.5mm]v2.south west) -- node[pos=1,anchor=south,yshift=1mm]{$s_2$, $v_2$} +(-6mm,0mm);
	\end{tikzpicture}
	\vskip-2mm%
	\caption{A vehicle following system.}
	\label{fig_example_platoon}
\end{center}
\end{figure}

\section{Illustrative example}\label{sec_example}
To illustrate the assume/guarantee reasoning enabled by contracts, consider the vehicle following system in Figure~\ref{fig_example_platoon}. We aim to show that a given controller for the follower vehicle (with index $2$) guarantees that the inter-vehicle distance to its predecessor satisfies a certain spacing policy, regardless of the behavior of the predecessor.

For the vehicle following system, the position $s_i$ and velocity $v_i$, $i\in\{1,2\}$, of both vehicles are regarded as the external variables, such that
\begin{align}
w^{\T} = \left[\begin{array}{cccc} s_1 & v_1 & s_2 & v_2 \end{array}\right],
\label{eqn_example_externalvariables}
\end{align}
and $\Wcal = \R^4$.
We assume that the exact dynamics of the first vehicle is unknown, but that the second vehicle satisfies
\begin{align}
\dot{s}_2 = v_2, \quad \dot{v}_2 = u_2,
\label{eqn_example_vehicledynamics}
\end{align}
and implements the controller
\begin{align}
u_2 = h^{-1}(v_1 - v_2) + kh^{-1}(s_1 - s_2 - hv_2),
\label{eqn_example_controller}
\end{align}
for some $h,k>0$. Now, the objective is to show that the controller (\ref{eqn_example_controller}) guarantees tracking of the so-called constant headway spacing policy (see, e.g., \cite{ioannou_1993}), i.e., that
\begin{align}
s_2(t) - s_1(t) = hv_2(t)
\label{eqn_example_spacingpolicy}
\end{align}
holds for all $t\in\R_+$. To guarantee this property using a contract $\con = (\ass,\gar)$, the requirement (\ref{eqn_example_spacingpolicy}) is captured by choosing the guarantees $\gar$ as in (\ref{eqn_sysgar}) with state $x^g = w^g$ according to (\ref{eqn_example_externalvariables}) and such that
\begin{align}
A^g = I, \quad\! G^g = I, \quad\! C^g = I, \quad\! H^g = \left[\begin{array}{cccc} -1 & 0 & 1 & h \end{array}\right].
\label{eqn_example_guarantees}
\end{align}
Note that the guarantees $\gar$ merely constrain the external variables $w^g$ through the constraint $H^g$ (corresponding to~(\ref{eqn_example_spacingpolicy})) and that the choice for $A^g$, $G^g$, and $C^g$ does not impose further restrictions. As a result, it is easy to show that the consistent subspace of $\gar$ is given as $\Vcal^{g,\star} = \ker H^g$.

Next, even though the exact dynamics of the first vehicle is unknown, it is safe to assume that its position $s_1$ and velocity $v_1$ satisfy the kinematic relation $\dot{s}_1 = v_1$. This can be expressed by choosing the assumptions $\ass$ in (\ref{eqn_sysass}) as
\begin{align}
A^a = \left[\begin{array}{cccc} 0 & 1 & 0 & 0 \\ 0 & 0 & 0 & 0 \\ 0 & 0 & 0 & 0 \\ 0 & 0 & 0 & 0 \end{array}\right]\!, \quad\!
G^a = \left[\begin{array}{ccc} 0 & 0 & 0 \\ 1 & 0 & 0 \\ 0 & 1 & 0 \\ 0 & 0 & 1 \end{array}\right]\!, \quad\!
C^a = I, \label{eqn_example_assumptions}
\end{align}
and with $H^a = 0$, corresponding to the state $x^a = w^a$. We stress that the form (\ref{eqn_example_assumptions}) characterizes the kinematic relation of the first vehicle, but does not constrain the behavior of the second vehicle (in terms of $s_2$ and $v_2$).

Given the contract $\con = (\ass,\gar)$ specified by (\ref{eqn_example_guarantees}) and (\ref{eqn_example_assumptions}), it remains to be shown that the vehicle (\ref{eqn_example_vehicledynamics}) with controller (\ref{eqn_example_controller}) satisfies the contract. To this end, the closed-loop system $\sys$ as in (\ref{eqn_sys}) (with state $x = w$) is represented as
\begin{align}
\!A = \left[\!\begin{array}{cccc} 0 & 0 & 0 & 0 \\ 0 & 0 & 0 & 0 \\ 0 & 0 & 0 & 1 \\ kh^{-1} & h^{-1} & -kh^{-1} & -k-h^{-1} \end{array}\!\right]\!, \,\,
G = \left[\begin{array}{cc} 1 & 0 \\ 0 & 1 \\ 0 & 0 \\ 0 & 0 \end{array}\right]\!,
\label{eqn_example_system}
\end{align}
and with $C = I$, $H = 0$. Similar to before, we stress that the system (\ref{eqn_example_system}) does not pose any restrictions on the behavior of the first vehicle, but only captures the dynamics (\ref{eqn_example_vehicledynamics})--(\ref{eqn_example_controller}).

In order to use Theorem~\ref{thm_contractimplementation} to verify contract implementation, the invariant subspace algorithm (see Remark~\ref{rem_geometriccontrol}) is used to compute the largest linear subspace $\Scal$ satisfying (\ref{eqn_thm_contractimplementation_state}) and (\ref{eqn_thm_contractimplementation_output}). This yields
\begin{align}
\Scal = \big\{ (x^a,x,x^g) \:\big|\: x^a = x = x^g,\, x^g \in \Vcal^{g,\star} \big\}.
\label{eqn_example_simulationrelation}
\end{align}
Then, after noting that the consistent subspace of the composition $\ass\comp\gar$ is given by $\Vcal^{\comp,\star} = \{ (x^a,x) \:|\: x^a = x \}$, it can be verified that $\Scal$ in (\ref{eqn_example_simulationrelation}) also satisfies (\ref{eqn_thm_contractimplementation_input}). Hence, $\Scal$ is a simulation relation of $\ass\comp\sys$ by $\gar$.

We note that, by definition of contract implementation and simulation (see Definition~\ref{def_simulation}), the simulation relation $\Scal$ needs to be full in order to guarantee contract implementation, see Theorem~\ref{thm_contractimplementation}. However, this is not the case as
\begin{align}
\pi_{\Xcal^a\times\Xcal}(\Scal) = \big\{ (x^a,x) \:\big|\: x^a = x,\, x\in\Vcal^{g,\star} \big\}
\subsetneq \Vcal^{\comp,\star}.
\end{align}
This limitation stems from the fact that satisfaction of the spacing policy (\ref{eqn_example_spacingpolicy}) is not guaranteed for initial conditions $x_0$ of $\sys$ that do not satisfy the constraint $H^gx_0 = 0$. Nonetheless, for initial conditions $x_0\in\ker H^g = \Vcal^{g,\star}$, the existence of the simulation relation $\Scal$ guarantees that the controlled vehicle (\ref{eqn_example_vehicledynamics})--(\ref{eqn_example_controller}) achieves tracking of the spacing policy (\ref{eqn_example_spacingpolicy}) (captured in $\gar$) \emph{for any} preceding vehicle that satisfies the kinematic relation (captured in $\ass$).

\begin{figure}
\begin{center}
	\vskip1.5mm%
	\begin{tikzpicture}
	\begin{axis}[
	height=27mm, 
	width=62mm,
	at={(0,0)}, scale only axis,
	xlabel={$t$ [s]}, xlabel style={yshift=1mm},
	ylabel={$v_1,v_2$ [m/s], $e$ [m]}, ylabel style={yshift=-1mm},
	xmin=0, xmax=15, ymin=-0.3, ymax=3.3,
	tick label style={font=\scriptsize},
	every axis/.append style={font=\footnotesize},
	legend pos=south east,
	]
	\addplot[gray,line width=0.75pt] table[x index=0,y index=1] {tikzdata_example_sim.dat};
	\addlegendentry{$v_1$};
	\addplot[black,line width=0.75pt] table[x index=0,y index=2] {tikzdata_example_sim.dat};
	\addlegendentry{$v_2$};
	\addplot[gray,line width=1pt] table[x index=0,y index=3] {tikzdata_example_sim.dat};
	\addlegendentry{$e$};
	\addplot[gray,line width=0.75pt,dashed] table[x index=0,y index=1] {tikzdata_example_sim2.dat};
	\addplot[black,line width=0.75pt,dashed] table[x index=0,y index=2] {tikzdata_example_sim2.dat};
	\addplot[gray,line width=1pt,dashed] table[x index=0,y index=3] {tikzdata_example_sim2.dat};
	\end{axis}	
	\end{tikzpicture}
	\vskip-3mm%
	\caption{Simulation of the model (\ref{eqn_example_vehicledynamics}), (\ref{eqn_example_controller}), and (\ref{eqn_example_firstvehicle}) for initial conditions $[\,1 \; 2 \; 0 \; 1\,]^{\T}\in\ker H^g$ (solid) and $[\,1 \; 2 \; 0.8 \; 1\,]^{\T}\notin\ker H^g$ (dashed) and parameters $h = 1$, $k = 0.25$, $c = 0.5$. Here, $e = -s_1 + s_2 + hv_1$ and the external disturbance is chosen as $d_1(t) = 1$ for $t\in[0,5)$ and $d_1(t) = 1 + \sin(t - 5)$ for $t\geq5$.}
	\label{fig_example_simulation}
\end{center}
\end{figure}
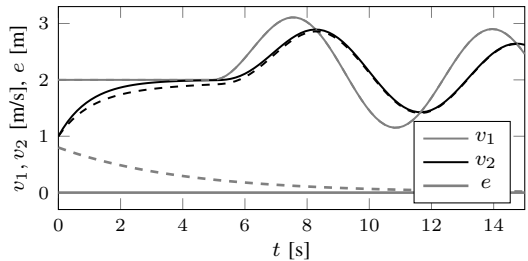

To illustrate this, let the dynamics of the first vehicle be given as
\begin{align}
\dot{s}_1 = v_1, \quad \dot{v}_1 = -cv_1 + d_1,
\label{eqn_example_firstvehicle}
\end{align}
for some constant $c>0$ and external disturbance $d_1$. It can be verified that the dynamics (\ref{eqn_example_firstvehicle}) is simulated by the assumptions $\ass$ in (\ref{eqn_example_assumptions}) (when (\ref{eqn_example_firstvehicle}) is appended with arbitrary dynamics for the second vehicle in a similar way as in (\ref{eqn_example_assumptions})). Hence, the satisfaction of the spacing policy is guaranteed for initial conditions $x_0\in\ker H^g$. This is confirmed by the results in Figure~\ref{fig_example_simulation}, which depicts (in solid lines) a time simulation of the model (\ref{eqn_example_vehicledynamics}), (\ref{eqn_example_controller}), and (\ref{eqn_example_firstvehicle}) for initial conditions $x_0\in\ker H^g$ and confirms that the spacing policy (\ref{eqn_example_spacingpolicy}) is satisfied for all time $t\geq0$ (even when the first vehicle is subject to time-varying disturbances). Finally, we note that it can in addition be shown that the subspace $\ker H^g$ is attractive, even though this is not a requirement in the contract $\con = (\ass,\gar)$. This is illustrated by a simulation in Figure~\ref{fig_example_simulation} as well (in dashed lines).

\section{Conclusions}\label{sec_conclusions}
Assume/guarantee contracts for dynamical systems are introduced in this paper as a means of characterizing control system specifications. For these contracts, a result is given that enables efficient algorithmic verification of contract satisfaction and the notion of refinement is introduced to allow for comparison of contracts.

We regard this work as the first step towards contract theory for continuous-time dynamical control systems. For such theory, tools for system composition are crucial (see \cite{sangiovanni-vincentelli_2012,benveniste_2018}) and future work will focus on this topic. The expression of relevant control specifications in terms of contracts will be a second topic of further research.

\section*{Appendix A. Proof of Theorem~\ref{thm_comp}}
The two statements are proven separately.

1) The result will be proven for the case $i=1$ by showing that the conditions in Lemma~\ref{lem_simulationrel} hold, after which the case $i=2$ follows similarly. Define the linear subspace $\Scal\subset\Xcal_1\times\Xcal_2\times\Xcal_1$ as
\begin{align}
\Scal = \big\{ (x^{\comp}_1,x^{\comp}_2,x_1) \:\big|\: x^{\comp}_1 = x_1,\, (x^{\comp}_1,x^{\comp}_2)\in\Vcal^{\comp,\star} \big\},
\label{eqn_thm_comp_proof_part1_S}
\end{align}
from which it can be observed that $\pi_{\Xcal_1\times\Xcal_2}(\Scal) = \Vcal^{\comp,\star}$. However, for $\Scal$ to be a full simulation relation, it also needs to satisfy $\pi_{\Xcal_1}(\Scal)\subset\Vcal_1^{\star}$ (see Definition~\ref{def_simulation}), where the projection $\pi_{\Xcal_1}$ is understood as
\begin{align}
\pi_{\Xcal_1}(\Scal) = \big\{ x_1 \:\big|\: &\exists(x^{\comp}_1,x^{\comp}_2)\in\Xcal_1\times\Xcal_2 \nonumber\\
&\text{such that } (x^{\comp}_1,x^{\comp}_2,x_1)\in\Scal \big\}.
\label{eqn_thm_comp_proof_part1_Sproj}
\end{align}
To show this, note that the definition of $\Scal$ in (\ref{eqn_thm_comp_proof_part1_S}) implies that (\ref{eqn_thm_comp_proof_part1_Sproj}) can be written as $\pi_{\Xcal_1}(\Scal) = \pi_{\Xcal_1}(\Vcal^{\comp,\star})$. By the defining properties of the consistent subspace (\ref{eqn_comp_consistentsubspace}) and the structure of the linear maps in (\ref{eqn_comp_matrices}), we obtain
\begin{align}
\begin{split}
A_1\pi_{\Xcal_1}(\Vcal^{\comp,\star}) &\subset \pi_{\Xcal_1}(\Vcal^{\comp,\star}) + \image G_1,\\
\pi_{\Xcal_1}(\Vcal^{\comp,\star}) &\subset \kernel H_1,
\label{eqn_thm_comp_proof_part1_step1}
\end{split}
\end{align}
and observe that this equals the defining equation of the consistent subspace of $\sys_1$ (see (\ref{eqn_consistentsubspace})). As a result, $\pi_{\Xcal_1}(\Scal) = \pi_{\Xcal_1}(\Vcal^{\comp,\star})\subset\Vcal_1^{\star}$ and the desired result is obtained.

Next, let $(x^{\comp}_1,x^{\comp}_2,x_1)\in\Scal$ and choose any $d^{\comp} = (d^{\comp}_1,d^{\comp}_2)\in\Dcal_1\times\Dcal_2$ such that $A^{\comp}x^{\comp} + G^{\comp}d^{\comp}\in\Vcal^{\comp,\star}$, where $x^{\comp} = (x^{\comp}_1,x^{\comp}_2)$. Then, by recalling that $x^{\comp}_1 = x_1$ due to the definition of $\Scal$ in (\ref{eqn_thm_comp_proof_part1_S}), it readily follows that
\begin{align}
\left[\begin{array}{c} A_1x^{\comp}_1 + G_1d^{\comp}_1 \\ A_2x^{\comp}_2 + G_2d^{\comp}_2 \\ A_1x_1 + G_1d_1 \end{array}\right]\in\Scal,
\label{eqn_thm_comp_proof_part1_step2}
\end{align}
for the choice $d_1 = d^{\comp}_1$. We stress that this choice also guarantees $A_1x_1 + G_1d_1\in\Vcal_1^{\star}$ as (\ref{eqn_thm_comp_proof_part1_step2}) implies $A_1x_1 + G_1d_1 \in \pi_{\Xcal_1}(\Scal)$ (recall the earlier result $\pi_{\Xcal_1}(\Scal)\subset\Vcal_1^{\star}$).

In addition, as $(x^{\comp}_1,x^{\comp}_2,x_1)\in\Scal$ implies $(x^{\comp}_1,x^{\comp}_2)\in\Vcal^{\comp,\star}$, the final algebraic constraint in (\ref{eqn_comp_sys}) gives $C_1x^{\comp}_1 = C_2x^{\comp}_2$, after which it is observed that
\begin{align}
C^{\comp}x^{\comp} = \tfrac{1}{2}C_1x^{\comp}_1 + \tfrac{1}{2}C_2x^{\comp}_2 = C_1x_1,
\label{eqn_thm_comp_proof_part1_step3}
\end{align}
where it is recalled that $x^{\comp}_1 = x_1$ on $\Scal$. As (\ref{eqn_thm_comp_proof_part1_step2}) and (\ref{eqn_thm_comp_proof_part1_step3}) correspond to the conditions of Lemma~\ref{lem_simulationrel}, the result (\ref{eqn_thm_comp_upperbound}) follows from application of this lemma.

2) In this part of the proof, the system $\sys$ will be presented as in (\ref{eqn_sysi}), but the subscripts will be omitted. Consequently, its consistent subspace is written as $\Vcal^*$. Let $\Scal_i\subset\Xcal\times\Xcal_i$ be full simulation relations of $\sys$ by $\sys_i$, $i\in\{1,2\}$ and define the linear subspace $\Scal\subset\Xcal\times\Xcal_1\times\Xcal_2$ as
\begin{align}
\Scal = \big\{ (x,x^{\comp}_1,x^{\comp}_2) \:\big|\: (x,x^{\comp}_1)\in\Scal_1,\, (x,x^{\comp}_2)\in\Scal_2 \big\}.
\label{eqn_thm_comp_proof_part2_S}
\end{align}
As both $\Scal_1$ and $\Scal_2$ are full simulation relations (i.e., they satisfy $\pi_{\Xcal}(\Scal_i) = \Vcal^{\star}$), we have $\pi_{\Xcal}(\Scal) = \Vcal^{\star}$. The desired condition $\pi_{\Xcal_1\times\Xcal_2}(\Scal)\subset\Vcal^{\comp,\star}$ will be shown later.

Choose any $(x,x^{\comp}_1,x^{\comp}_2)\in\Scal$ and let $d\in\Dcal$ be such that $Ax + Bd\in\Vcal^{\star}$. Now, consider
\begin{align}
\left[\begin{array}{c} Ax + Gd \\ A_1x^{\comp}_1 + G_1d^{\comp}_1 \\ A_2x^{\comp}_2 + G_2d^{\comp}_2 \end{array}\right]
= \left[\begin{array}{c} s \\ s^{\comp}_1 \\ s^{\comp}_2 \end{array}\right].
\label{eqn_thm_comp_proof_part2_step1}
\end{align}
Since $(x,x^{\comp}_1)\in\Scal_1$ by (\ref{eqn_thm_comp_proof_part2_S}) and by recalling that $\Scal_1$ is a simulation relation of $\sys$ by $\sys_1$, there exists $d^{\comp}_1\in\Dcal_1$ such that $(s,s^{\comp}_1)\in\Scal_1$. Similarly, $d^{\comp}_2\in\Dcal_2$ can be chosen to satisfy $(s,s^{\comp}_2)\in\Scal_2$, after which
\begin{align}
(s,s^{\comp}_1,s^{\comp}_2)\in\Scal.
\label{eqn_thm_comp_proof_part2_step2}
\end{align}

Next, $(x,x^{\comp}_i)\in\Scal_i$ for $i\in\{1,2\}$ lead to $Cx = C_1x^{\comp}_1$ and $Cx = C_1x^{\comp}_1$, respectively, such that
\begin{align}
Cx = \tfrac{1}{2}C_1x^{\comp}_1 + \tfrac{1}{2}C_2x^{\comp}_2 = C^{\comp}x^{\comp}
\label{eqn_thm_comp_proof_part2_step3}
\end{align}
for all $(x,x^{\comp}_1,x^{\comp}_2)\in\Scal$ and where $x^{\comp} = (x^{\comp}_1,x^{\comp}_2)$ and $C^{\comp}$ are the state and output map of $\sys_1\comp\sys_2$, respectively. Even though (\ref{eqn_thm_comp_proof_part2_step2}) and (\ref{eqn_thm_comp_proof_part2_step3}) resemble the conditions of Lemma~\ref{lem_simulationrel}, it needs to be shown that $\pi_{\Xcal_1\times\Xcal_2}(\Scal)\subset\Vcal^{\comp,\star}$ before this lemma can be applied.

To this end, observe that (\ref{eqn_thm_comp_proof_part2_step1}) and (\ref{eqn_thm_comp_proof_part2_step2}) imply
\begin{align}
\left[\begin{array}{ccc} A & 0 & 0 \\ 0 & A_1 & 0 \\ 0 & 0 & A_2 \end{array}\right]\Scal \subset \Scal + 
\image \left[\begin{array}{ccc} G & 0 & 0 \\ 0 & G_1 & 0 \\ 0 & 0 & G_2 \end{array}\right],
\end{align}
from which it can be concluded that
\begin{align}
A^{\comp}\pi_{\Xcal_1\times\Xcal_2}(\Scal) \subset \pi_{\Xcal_1\times\Xcal_2}(\Scal) + \image G^{\comp}.
\label{eqn_thm_comp_proof_part2_step4}
\end{align}
Here, recall that $A^{\comp}$ and $G^{\comp}$ represent $\sys_1\comp\sys_2$ as defined in (\ref{eqn_comp_matrices}). Now, take any $(x^{\comp}_1,x^{\comp}_2)\in\pi_{\Xcal_1\times\Xcal_2}(\Scal)$ and let $x\in\Xcal$ be such that $(x,x^{\comp}_1,x^{\comp}_2)\in\Scal$. By the reasoning above (\ref{eqn_thm_comp_proof_part2_step3}), we have $Cx = C_1x^{\comp}_1$ and $Cx = C_2x^{\comp}_2$, leading to $C_1x^{\comp}_1=C_2x^{\comp}_2$. Stated differently, the final algebraic constraint in (\ref{eqn_sysi}) is satisfied. In addition, as $(x,x^{\comp}_i)\in\Scal_i$ and $\pi_{\Xcal_i}(\Scal_i)\subset\Vcal_i^{\star}$ by definition, it holds that $H_ix^{\comp}_i = 0$ for $i\in\{1,2\}$ (see the definition of $\Vcal_i^{\star}$ in (\ref{eqn_consistentsubspace})). Combining the above observations leads to the conclusion
\begin{align}
\pi_{\Xcal_1\times\Xcal_2}(\Scal) \subset \kernel H^{\comp},
\label{eqn_thm_comp_proof_part2_step5}
\end{align}
with $H^{\comp}$ as in (\ref{eqn_comp_matrices}). Now, a comparison of (\ref{eqn_thm_comp_proof_part2_step4}) and (\ref{eqn_thm_comp_proof_part2_step5}) with (\ref{eqn_comp_consistentsubspace}) shows that $\pi_{\Xcal_1\times\Xcal_2}(\Scal)\subset\Vcal^{\comp,\star}$ as $\Vcal^{\comp,\star}$ is the \emph{largest} subspace satisfying (\ref{eqn_comp_consistentsubspace}). Note also that this result implies that the choice of $(d^{\comp}_1,d^{\comp}_2)$ below (\ref{eqn_thm_comp_proof_part2_step1}) is such that $(s^{\comp}_1,s^{\comp}_2)\in\Vcal^{\comp,\star}$.

At this point, all conditions of Lemma~\ref{lem_simulationrel} are satisfied for the full simulation relation $\Scal$ in (\ref{eqn_thm_comp_proof_part2_S}). Consequently, the application of this lemma shows the result (\ref{eqn_thm_comp_infimum}).

\section*{Appendix B. Proof of Theorem~\ref{thm_contractimplementation}}
\textit{If.} If will be shown that (\ref{eqn_thm_contractimplementation_state})--(\ref{eqn_thm_contractimplementation_output}) imply conditions (\ref{eqn_lem_simulationrel_state}) and (\ref{eqn_lem_simulationrel_output}) for $\sys_1 = \ass\comp\sys$ and $\sys_2 = \gar$, after which the result follows from Lemma~\ref{lem_simulationrel}.

First, from (\ref{eqn_thm_contractimplementation_state}) and (\ref{eqn_thm_contractimplementation_output}) as well as the definition of the projection $\pi_{\Xcal^a\times\Xcal}(\Scal)$, it holds that
\begin{align}
A^{\comp}\pi_{\Xcal^a\times\Xcal}(\Scal) &\subset \pi_{\Xcal^a\times\Xcal}(\Scal) + \image G^{\comp}, \\
\pi_{\Xcal^a\times\Xcal}(\Scal) &\subset \kernel H^{\comp}.
\end{align}
By comparing the above to the definition of the consistent subspace in (\ref{eqn_consistentsubspace}) and recalling that $\Vcal^{\comp,\star}$ is the \emph{largest} subspace satisfying (\ref{eqn_consistentsubspace}), it is immediate that $\pi_{\Xcal^a\times\Xcal}(\Scal)\subset\Vcal^{\comp,\star}$. The inclusion $\pi_{\Xcal^g}(\Scal)\subset\Vcal^{g,\star}$ can be shown analogously.

To show that (\ref{eqn_lem_simulationrel_state}) holds, take any $(x^{\comp},x^g)\in\Scal$, with $x^{\comp}\in\Xcal^a\times\Xcal$ the state of $\ass\comp\sys$. Then, by (\ref{eqn_thm_contractimplementation_state}), there exist $(s^{\comp},s^g)\in\Scal$ and $(v^{\comp},v^g)\in\Dcal^{\comp}\times\Dcal^g$ (with $\Dcal^{\comp} = \Dcal^a\times\Dcal$) such that
\begin{align}
\left[\begin{array}{c} A^{\comp}x^{\comp} \\ A^gx^g \end{array}\right] =
\left[\begin{array}{c} s^{\comp} \\ s^g \end{array}\right] - \left[\begin{array}{c} G^{\comp}v^{\comp} \\ G^gv^g \end{array}\right].
\label{eqn_thm_simulationrel_proof_step1}
\end{align}
In addition, take any $d^{\comp}\in\Dcal^{\comp}$ such that $A^{\comp}x^{\comp} + G^{\comp}d^{\comp}\in\Vcal^{\comp,\star}$ and note that
\begin{align}
G^{\comp}(d^{\comp} - v^{\comp}) &= \big(A^{\comp}x^{\comp} + G^{\comp}d^{\comp}\big) - s^{\comp},
\label{eqn_thm_simulationrel_proof_step2}
\end{align}
with $s^{\comp}$ satisfying (\ref{eqn_thm_simulationrel_proof_step1}).
As $A^{\comp}x^{\comp} + G^{\comp}d^{\comp} \in \Vcal^{\comp,\star}$ by choice and $s^{\comp} \in \pi_{\Xcal^a\times\Xcal}(\Scal) \subset\Vcal^{\comp,\star}$ by the earlier result on $\Scal$, it follows that
\begin{align}
G(d^{\comp} - v^{\comp}) \in \image G^{\comp} \cap \Vcal^{\comp,\star}.
\end{align}
Then, application of (\ref{eqn_thm_contractimplementation_input}) guarantees the existence of $(\tilde{s}^{\comp},\tilde{s}^g)\in\Scal$ such that
\begin{align}
\left[\begin{array}{c} G(d^{\comp} - v^{\comp}) \\ 0 \end{array}\right] =
\left[\begin{array}{c} \tilde{s}^{\comp} \\ \tilde{s}^g \end{array}\right] + \left[\begin{array}{c} 0 \\ G^g\tilde{v}^g \end{array}\right],
\label{eqn_thm_simulationrel_proof_step3}
\end{align}
for some $\tilde{v}^g$ satisfying $G^g\tilde{v}^g\in \image G^g \cap \Vcal^{g,\star}$. Now, for $d^g = v^g + \tilde{v}^g$, the result
\begin{align}
\left[\begin{array}{c} A^{\comp}x^{\comp} + G^{\comp}d^{\comp} \\ A^gx^g + G^gd^g \end{array}\right] &= 
\left[\begin{array}{c} A^{\comp}x^{\comp} + G^{\comp}v^{\comp} + G^{\comp}(d^{\comp} - v^{\comp}) \\ A^gx^g + G^gv^g + G^g\tilde{v}^g \end{array}\right] \nonumber\\
&=
\left[\begin{array}{c} s^{\comp} + \tilde{s}^{\comp} \\ s^g + \tilde{s}^g \end{array}\right] \in \Scal
\label{eqn_thm_simulationrel_proof_step4}
\end{align}
is obtained by applying (\ref{eqn_thm_simulationrel_proof_step1}) and (\ref{eqn_thm_simulationrel_proof_step3}). Thus, for any $(x^{\comp},x^g)\in\Scal$ and any $d^{\comp}\in\Dcal^{\comp}$ such that $A^{\comp}x^{\comp} + G^{\comp}d^{\comp}\in\Vcal^{\comp,\star}$, we have constructed a $d^g\in\Dcal^g$ such that (\ref{eqn_thm_simulationrel_proof_step4}) holds. Note, in addition, that $A^gx^g + G^gd^g = s^g + \tilde{s}^g\in\pi_{\Xcal^g}(\Scal)\subset\Vcal^{g,\star}$, implying that condition \ref{lem_simulationrel_item1} in Lemma~\ref{lem_simulationrel} is satisfied.

It remains to be shown that (\ref{eqn_lem_simulationrel_output}) holds, but this is immediate from (\ref{eqn_thm_contractimplementation_output}).

\textit{Only if.} Let $\Scal$ satisfy $\pi_{\Xcal^a\times\Xcal}(\Scal)\subset\Vcal^{\comp,\star}$ and $\pi_{\Xcal^g}(\Scal)\subset\Vcal^{g,\star}$ as well as the conditions of Lemma~\ref{lem_simulationrel}. It is clear that (\ref{eqn_lem_simulationrel_state}) implies (\ref{eqn_thm_contractimplementation_state}). Next, consider (\ref{eqn_lem_simulationrel_state}) for $(x^{\comp},x^g) = 0$ and take any $d^{\comp}$ such that $G^{\comp}d^{\comp}\in\Vcal^{\comp,\star}$. Then, there exists $(s^{\comp},s^g)\in\Scal$ and $d^g\in\Dcal^g$ such that $G^gd^g\in\Vcal^{g,\star}$ and
\begin{align}
\left[\begin{array}{c} G^{\comp}d^{\comp} \\ 0 \end{array}\right] = 
\left[\begin{array}{c} s^{\comp} \\ s^g \end{array}\right] + \left[\begin{array}{c} 0 \\ G^gd^g \end{array}\right],
\end{align}
from which (\ref{eqn_thm_contractimplementation_input}) follows. Finally, (\ref{eqn_thm_contractimplementation_output}) follows from $\pi_{\Xcal^a\times\Xcal}(\Scal)\subset\Vcal^{\comp,\star}$ and $\pi_{\Xcal^g}(\Scal)\subset\Vcal^{g,\star}$ (recall the definition of the consistent subspace in (\ref{eqn_consistentsubspace})) and (\ref{eqn_lem_simulationrel_output}).

\bibliographystyle{plain}
\bibliography{main}

\end{document}

%% file: header_ieeeconf.tex
\usepackage{amsmath}
\usepackage{amssymb}
\usepackage{bm}

\usepackage[english]{babel}

\usepackage{graphicx}

\usepackage{tikz}
\usetikzlibrary[positioning,arrows,backgrounds,calc]
\usepackage{pgfplots}
\pgfplotsset{compat=newest}

\newtheorem{theorem}{Theorem}
\newtheorem{lemma}[theorem]{Lemma}

\newtheorem{definition}{Definition}

\newtheorem{remarkenv}{Remark}
\newenvironment{remark}{\begin{remarkenv}}{\hfill\raisebox{0.5mm}[0cm][0cm]{$\lhd$}\end{remarkenv}}

\newcommand{\sortbib}[1]{}

\usepackage{dsfont}
\newcommand{\R}{\ensuremath{\mathds{R}}} 

\newcommand{\T}{\ensuremath{\mathrm{T}}}                          

\usepackage{xspace}

\newcommand{\Acal}{\ensuremath{\mathcal{A}}}

\newcommand{\Ccal}{\ensuremath{\mathcal{C}}}
\newcommand{\Dcal}{\ensuremath{\mathcal{D}}}
\newcommand{\Ecal}{\ensuremath{\mathcal{E}}}

\newcommand{\Gcal}{\ensuremath{\mathcal{G}}}

\newcommand{\Scal}{\ensuremath{\mathcal{S}}}

\newcommand{\Vcal}{\ensuremath{\mathcal{V}}}
\newcommand{\Wcal}{\ensuremath{\mathcal{W}}}
\newcommand{\Xcal}{\ensuremath{\mathcal{X}}}
\newcommand{\Ycal}{\ensuremath{\mathcal{Y}}}

\newcommand{\sys}{\ensuremath{\bm{\Sigma}}}